\documentclass[paper,a4paper]{IEEEtran}

\usepackage{amsthm}
\usepackage{amsmath}
\usepackage{amsfonts,amssymb,amsmath}
\usepackage{graphicx}
\usepackage{epsfig}
\usepackage{caption}
\usepackage[applemac]{inputenc}
\usepackage{latexsym}
\usepackage{dsfont}
\usepackage{tikz}
\usetikzlibrary{arrows,shapes,backgrounds,patterns,decorations.pathreplacing}

\newtheorem{definition}{Definition}
\newtheorem{corollary}{Corollary}
\newtheorem{proposition}{Proposition}
\newtheorem{theorem}{Theorem}
\newtheorem{lemma}{Lemma}

\newtheorem{remark}{Remark}

\newtheorem{example}{Example}

\begin{document}

\sloppy

\title{On the Polarization Levels of Automorphic-Symmetric Channels}

\author{
  \IEEEauthorblockN{Rajai Nasser\\}
  \IEEEauthorblockA{
Email: rajai.nasser@alumni.epfl.ch} 
}



\maketitle

\begin{abstract}
It is known that if an Abelian group operation is used in an Ar{\i}kan-style construction, we have multilevel polarization where synthetic channels can approach intermediate channels that are neither almost perfect nor almost useless. An open problem in polarization theory is to determine the polarization levels of a given channel. In this paper, we discuss the polarization levels of a family of channels that we call automorphic-symmetric channels. We show that the polarization levels of an automorphic-symmetric channel are determined by characteristic subgroups. In particular, if the group that is used does not contain any non-trivial characteristic subgroup, we only have two-level polarization to almost perfect and almost useless channels.
\end{abstract}

\section{Introduction}
Polar codes are the first family of low-complexity capacity-achieving codes. Polar codes were first introduced by Ar{\i}kan for binary-input channels \cite{Arikan}. The construction of polar codes relies on a phenomenon that is called \emph{polarization}: A collection of independent copies of the channel is transformed into a collection of synthetic channels that are almost perfect or almost useless.

The transformation of Ar{\i}kan for binary-input channels uses the XOR operation. The polarization phenomenon was generalized to channels with non-binary input by replacing the XOR operation with a binary operation on the input-alphabet \cite{SasogluTelAri,SasS,ParkBarg,SahebiPradhan,RajTelA,RajErgI,RajErgII}. Note that if the input alphabet size is not prime, we may have multilevel polarization where the synthetic channels can polarize to intermediate channels that are neither almost perfect nor almost useless. In this paper, we are interested in the multilevel polarization phenomenon when an Abelian group operation is used. More precisely, we are interested in determining the polarization levels of a family of channels that we call automorphic-symmetric channels.

In Section \ref{sec2}, we introduce the preliminaries of this paper. In Section \ref{sec3} we introduce $\mathcal{H}$-polarizing and strongly-polarizing families of channels. We show that if $\mathcal{W}$ is an $\mathcal{H}$-polarizing family of channels, then the polarization levels of every channel in $\mathcal{W}$ are determined by subgroups in $\mathcal{H}$. In Section \ref{sec4} we show that the family of $q$-ary erasure channels is strongly polarizing. This implies that every $q$-ary erasure channel polarizes to almost perfect and almost useless channels. In Section \ref{sec5} we introduce $q$-symmetric channels and generalized $q$-symmetric channels. $q$-symmetric channels generalize binary symmetric channels to arbitrary input alphabets. Generalized $q$-symmetric channels are a generalization of binary-input memoryless symmetric-output (BMS) channels. In Section \ref{sec6}, we introduce the family of automorphic-symmetric channels. We show that generalized $q$-symmetric channels are automorphic-symmetric. We show that the polarization levels of an automorphic-symmetric channel are determined by characteristic subgroups. This implies that if the group that is used does not contain any non-trivial characteristic subgroup, we only have two-level polarization to almost perfect and almost useless channels.

\section{Preliminaries}

\label{sec2}

Throughout this paper, $(G,+)$ denotes a fixed finite Abelian group, and $q=|G|$ denotes its size.

Let $\mathcal{Y}$ be a finite set. We write $W:G\longrightarrow \mathcal{Y}$ to denote a discrete memoryless channel (DMC) of input alphabet $G$ and output alphabet $\mathcal{Y}$. We write $I(W)$ to denote the symmetric capacity\footnote{The symmetric capacity is the mutual information between a uniformly distributed input and its corresponding output.} of $W$.

Define the two channels $W^-:G\longrightarrow\mathcal{Y}^2$ and $W^+:G\longrightarrow\mathcal{Y}^2\times G$ as follows:
$$W^-(y_1,y_2|u_1)=\frac{1}{q}\sum_{u_2\in G}W(y_1|u_1+u_2)W(y_2|u_2),$$
and
$$W^+(y_1,y_2,u_1|u_2)=\frac{1}{q}W(y_1|u_1+u_2)W(y_2|u_2).$$
Furthermore, for every $n\geq 1$ and every $s=(s_1,\ldots,s_n)\in\{-,+\}^n$, define the channel $W^s=(\ldots(W^{s_1})^{s_2}\ldots)^{s_n}$.

Now let $H$ be a subgroup of $(G,+)$. We denote by $G/H$ the quotient group of $G$ by $H$. Define the channel $W[H]:G/H\longrightarrow \mathcal{Y}$ as follows:
$$W[H](y|A)=\frac{1}{|A|}\sum_{x\in A}W(y|x)=\frac{1}{|H|}\sum_{x\in A}W(y|x).$$
It is easy to see that if $X$ is a uniformly distributed random variable in $G$ and $Y$ is the output of $W$ when $X$ is the input, then $I(W[H])=I(X\bmod H;Y)$.

\begin{definition}
Let $\delta>0$. A channel $W:G\longrightarrow\mathcal{Y}$ is said to be $\delta$-determined by a subgroup $H$ of $G$ if $\big|I(W)-\log|G/H|\big|<\delta$ and $\big|I(W[H])-\log|G/H|\big|<\delta$.
\end{definition}

The inequalities $\big|I(W)-\log|G/H|\big|<\delta$ and $\big|I(W[H])-\log|G/H|\big|<\delta$ can be interpreted as follows: Let $X$ be a uniformly distributed random variable in $G$ and let $Y$ be the output of the channel $W$ when $X$ is the input. If $\delta>0$ is small and $\big|I(W[H])-\log|G/H|\big|<\delta$, then from $Y$ we can determine $X\bmod H$ with high probability. If we also have $\big|I(W)-\log|G/H|\big|<\delta$, then $X\bmod H$ is almost the only information about $X$ which can be reliably deduced from $Y$. This is why we can say that if $W$ is $\delta$-determined by $H$ for a small $\delta$, then $W$ behaves similarly to a deterministic homomorphism channel projecting its input onto $G/H$.

It was proven in \cite{RajTelA} that as the number $n$ of polarization steps becomes large, the synthetic channels $(W^s)_{s\in\{-,+\}^n}$ polarize to deterministic homomorphism channels projecting their input onto quotient groups. More precisely, for every $\delta>0$, we have
\begin{align*}
\lim_{n\to\infty}\frac{1}{2^n}\Big|\Big\{s\in\{-,+\}^n:\;&\exists H_s\;\text{a subgroup of}\;G,\\
&W^s\;\text{is}\;\delta\text{-determined by}\;H_s\Big\}\Big|=1.
\end{align*}

\section{$\mathcal{H}$-Polarizing Families of Channels}

\label{sec3}

\begin{definition}
Let $\mathcal{H}$ be a set of subgroups of $(G,+)$. We say that a channel $W:G\longrightarrow\mathcal{Y}$ \emph{$\mathcal{H}$-polarizes} if for every $\delta>0$, we have
\begin{align*}
\lim_{n\to\infty}\frac{1}{2^n}\Big|\Big\{s\in\{-,+\}^n:\;&\exists H_s\in\mathcal{H},\\
&W^s\;\text{is}\;\delta\text{-determined by}\;H_s\Big\}\Big|=1.
\end{align*}

If $\mathcal{H}=\{\{0\},G\}$ and $W$ $\mathcal{H}$-polarizes, we say that $W$ \emph{strongly polarizes}.
\end{definition}
If $W$ $\mathcal{H}$-polarizes, then the levels of polarization are determined by subgroups in $\mathcal{H}$. $W$ strongly polarizes if and only if its synthetic channels $(W^s)_{s\in\{-,+\}^n}$ polarize only to almost useless and almost perfect channels.

Let $W:G\longrightarrow\mathcal{Y}$ be a given channel, and assume that after simulating enough polarization steps, we are convinced that $W$ $\mathcal{H}$-polarizes for some family $\mathcal{H}$ of subgroups. How can we prove that this is indeed the case? Characterizing $\mathcal{H}$-polarizing channels seems to be very difficult. In this paper, we aim to provide sufficient conditions for $\mathcal{H}$-polarization.

Our approach to show the $\mathcal{H}$-polarization of a channel, is to show that it belongs to what we call \emph{$\mathcal{H}$-polarizing family of channels}:

\begin{definition}
A family $\mathcal{W}$ of channels with input alphabet $G$ is said to be \emph{$\mathcal{H}$-polarizing} if it satisfies the following conditions:
\begin{itemize}
\item If $W\in\mathcal{W}$, then $W^-\in\mathcal{W}$ and $W^+\in\mathcal{W}$.\footnote{This implies that $W^s\in\mathcal{W}$ for every $n\geq 1$ and every $s\in\{-,+\}^n$.}
\item There exists $\delta_{\mathcal{W},\mathcal{H}}>0$ such that $\mathcal{W}$ does not contain any channel that is $\delta_{\mathcal{W},\mathcal{H}}$-determined by a subgroup other than those in $\mathcal{H}$.
\end{itemize}
\end{definition}

\begin{proposition}
\label{propMain}
Let $\mathcal{H}$ be a family of subgroups and let $\mathcal{W}$ be an $\mathcal{H}$-polarizing family of channels. Every channel $W\in\mathcal{W}$ $\mathcal{H}$-polarizes.
\end{proposition}
\begin{proof}
Fix $W\in\mathcal{W}$ and let $0<\delta<\delta_{\mathcal{W},\mathcal{H}}$. For every $n\geq 1$, define
\begin{align*}
A_{n,\delta}=\Big\{s\in\{-,+\}^n:\;\exists H_s\;&\text{a subgroup of}\;G,\\
&W^s\;\text{is}\;\delta\text{-determined by}\;H_s\Big\}.
\end{align*}

We have $\displaystyle\lim_{n\rightarrow\infty}\frac{1}{2^n}|A_{n,\delta}|=1$. Let $s\in A_{n,\delta}$. There exists a subgroup $H_s$ of $G$ such that $W^s$ is $\delta$-determined by $H_s$. Since $W\in\mathcal{W}$, then $W^s\in\mathcal{W}$, which implies that $W^s$ cannot be $\delta$-determined by a subgroup other than those in $\mathcal{H}$. Therefore, $H_s\in\mathcal{H}$. We conclude that
\begin{align*}
\lim_{n\to\infty}\frac{1}{2^n}\Big|\Big\{s\in\{-,+\}^n:\;&\exists H_s\in\mathcal{H},\\
&W^s\;\text{is}\;\delta\text{-determined by}\;H_s\Big\}\Big|=1,
\end{align*}
which means that $W$ $\mathcal{H}$-polarizes.
\end{proof}

\section{$q$-ary Erasure Channels}
\label{sec4}

Our first example of a strongly polarizing family of channels is the family of $q$-ary erasure channels.

\begin{definition}
Let $e$ be a symbol that does not belong to $G$. We say that a channel $W:G\longrightarrow G\cup\{e\}$ is a $q$-ary erasure channel with parameter $\epsilon$ (denoted $W=qEC(\epsilon)$) if
$$W(y|x)=\begin{cases}1-\epsilon\quad&\text{if}\;y=x,\\\epsilon\quad&\text{if}\;y=e,\\0\quad\text{otherwise}.\end{cases}$$
\end{definition}

We also call $q$-ary erasure channel any channel that is \emph{equivalent} to $qEC(\epsilon)$ in the following sense:

\begin{definition}
A channel $W:G\longrightarrow\mathcal{Y}$ is said to be \emph{degraded} from another channel $W':G\longrightarrow\mathcal{Y}'$ if there exists a channel $V':\mathcal{Y}'\longrightarrow\mathcal{Y}$ such that
$$W(y|x)=\sum_{y'\in\mathcal{Y}'}W'(y'|x)V'(y|y').$$
$W$ and $W'$ are said to be \emph{equivalent} if they are degraded from each other.
\end{definition}

Denote by $\mathcal{W}_{qEC}$ the family of all $q$-ary erasure channels.

\begin{lemma}
If $W\in \mathcal{W}_{qEC}$, then $W^-\in\mathcal{W}_{qEC}$ and $W^+\in\mathcal{W}_{qEC}$.
\label{lemqEC1}
\end{lemma}
\begin{proof}
It is easy to check that if $W$ is equivalent to $qEC(\epsilon)$, then $W^-$ is equivalent to $qEC(2\epsilon-\epsilon^2)$ and $W^+$ is equivalent to $qEC(\epsilon^2)$.
\end{proof}

\begin{lemma}
\label{lemqEC2}
There exists $\delta_{qEC}>0$ such that there is no $q$-ary erasure channel that is $\delta_{qEC}$-determined by a non-trivial\footnote{The trivial subgroups of $(G,+)$ are $\{0\}$ and $G$.} subgroup.
\end{lemma}
\begin{proof}
Define $\displaystyle\delta_{qEC}=\frac{(\log 2)^2}{\log (2q)}$. Let $W=qEC(\epsilon)$ be a $q$-ary erasure channel and assume there exists a non-trivial subgroup $H$ of $G$ such that $\big|I(W[H])-\log|G/H|\big|<\delta_{qEC}$. It is easy to check that $W[H]$ is a $\displaystyle\frac{q}{|H|}$-erasure channel of input alphabet $G/H$ and of parameter $\epsilon$. Moreover, we have $I(W[H])=(\log|G/H|)(1-\epsilon)$. Now since $\big|I(W[H])-\log|G/H|\big|<\delta_{qEC}=\frac{(\log 2)^2}{\log (2q)}$, we have
$$\resizebox{0.48\textwidth}{!}{$\displaystyle\epsilon\log|G/H| <\frac{(\log 2)^2}{\log (2q)}\;\stackrel{(a)}{\Rightarrow}\; \epsilon<\frac{(\log2)^2}{\log (2q)\log|G/H|}\leq\frac{\log 2}{\log (2q)},$}$$
where (a) follows from the fact that $H$ is non-trivial (and hence $|G/H|\geq 2$). Thus,
\begin{align*}
I(W)-\log|G/H|&=(\log q)(1-\epsilon)-\log\frac{q}{|H|}\\
&= \log |H| -\epsilon\log q\\
&\stackrel{(a)}{\geq} \log 2 -  \epsilon\log q\\
&> \log 2 - 
\frac{(\log q)(\log2)}{\log 2 + \log q}\\
&=\frac{(\log 2)^2}{\log (2q)}=\delta_{qEC},
\end{align*}
where (a) follows from the fact that $H$ is non-trivial (and hence $|H|\geq 2$). Therefore, we cannot have $\big|I(W)-\log|G/H|\big|<\delta_{qEC}$.

We conclude that if $W$ is a channel with input alphabet $G$ such that there exists a non-trivial subgroup $H$ of $G$ satisfying $\big|I(W)-\log|G/H|\big|<\delta_{qEC}$ and $\big|I(W[H])-\log|G/H|\big|<\delta_{qEC}$, then $W\notin\mathcal{W}_{qEC}$.
\end{proof}

\begin{proposition}
\label{propqEC}
$\mathcal{W}_{qEC}$ is a strongly polarizing family of channels.
\end{proposition}
\begin{proof}
The proposition follows from Lemmas \ref{lemqEC1} and \ref{lemqEC2}.
\end{proof}

\begin{corollary}
Every $q$-ary erasure channel with input alphabet $G$ strongly polarizes.
\end{corollary}
\begin{proof}
The corollary follows from Propositions \ref{propMain} and \ref{propqEC}
\end{proof}

\section{$q$-Symmetric Channels and Generalized $q$-Symmetric Channels}

\label{sec5}

\begin{definition}
Let $\displaystyle 0\leq\epsilon\leq\frac{1}{q-1}$. The $q$-symmetric channel of parameter $\epsilon$ (denoted $qSC(\epsilon)$) is the channel $W:G\longrightarrow G$ defined as
$$W(y|x)=\begin{cases}1-(q-1)\epsilon\quad&\text{if}\;y=x,\\\epsilon\quad&\text{otherwise.}\end{cases}$$
\end{definition}

$q$-symmetric channels generalize the binary symmetric channels to non-binary input alphabets.

We are interested in showing the strong polarization of $q$-symmetric channels. More generally, we are interested in showing the strong polarization of a more general family of channels:

\begin{definition}
\label{defGenqSym}
We say that a channel $W:G\longrightarrow\mathcal{Y}$ is a generalized $q$-symmetric channel if there exist a set $\mathcal{Y}_W$ and a bijection $\pi_W:G\times\mathcal{Y}_W\rightarrow\mathcal{Y}$ such that:
\begin{itemize}
\item There exists a mapping $p_W:\mathcal{Y}_W\rightarrow[0,1]$ such that $\displaystyle\sum_{y'\in\mathcal{Y}_W}p_W(y')=1$, i.e., $p_W$ is a probability distribution on $\mathcal{Y}_W$.
\item For every $y'\in\mathcal{Y}_W$, there exists $\displaystyle 0\leq\epsilon_{y'}\leq\frac{1}{q-1}$ such that for every $x,x'\in G$, we have:
\begin{align*}
W(\pi_W&(x',y')|x)\\
&=\begin{cases}p_W(y')\cdot(1-(q-1)\epsilon_{y'})\quad&\text{if}\;x'=x,\\p_W(y')\cdot\epsilon_{y'}\quad&\text{otherwise}.\end{cases}
\end{align*}
\end{itemize}
\end{definition}

Generalized $q$-symmetric channels generalize binary memoryless symmetric-output (BMS) channels.

\begin{example}
$qSC(\epsilon)$ is a generalized $q$-symmetric channel: Let $\mathcal{Y}_W=\{0\}$, define $\pi_W:G\times\mathcal{Y}_W\rightarrow G$ as $\pi_W(x,0)=x$, and define $p_W(0)=1$.
\end{example}

\begin{example}
Every $q$-ary erasure channel is equivalent to a generalized $q$-symmetric channel.
\end{example}

\begin{remark}
A generalized $q$-symmetric channel can be 	thought of as a combination of $q$-symmetric channels indexed by $y'\in\mathcal{Y}_W$:
\begin{itemize}
\item The channel picks $y'\in\mathcal{Y}_W$ with probability $p_W(y')$ and independently from the input.
\item The channel sends the input $x$ through a channel $qSC(\epsilon_{y'})$ and obtains $x'$.
\item The channel output is $y=\pi_W(x',y')$.
\end{itemize}
Since $\pi_W$ is a bijection, the receiver can recover $(x',y')$ from $y$. In other words, the receiver knows which $qSC$ from the collection $\{qSC(\epsilon_{y'}):\;y'\in\mathcal{Y}_W\}$ was used. Moreover, the receiver knows the $qSC$ output $x'$.
\end{remark}

The reader can check that if $W$ is a generalized $q$-symmetric channel, then $W^-$ is a generalized $q$-symmetric channel as well. Unfortunately, $W^+$ is not necessarily a generalized $q$-symmetric channel. Therefore, generalized $q$-symmetric channels do not form a strongly polarizing family of channels. In the next section, we will see that under some condition on the group $(G,+)$, generalized $q$-symmetric channels form a subfamily of a strongly polarizing family of channels.

\section{Automorphic-Symmetric Channels}
\label{sec6}
\begin{definition}
An automorphism of $G$ is an isomorphism\footnote{An isomorphism is a bijective homomorphism.} from $G$ to itself.
\end{definition}

\begin{definition}
A channel $W:G\longrightarrow\mathcal{Y}$ is said to be \emph{automorphic-symmetric with respect to $(G,+)$} if for every automorphism $f:G\rightarrow G$ there exists a bijection $\pi_f:\mathcal{Y}\rightarrow\mathcal{Y}$ such that $W(\pi_f(y)|f(x))=W(y|x)$.
\end{definition}

\begin{example}
If $G\equiv\mathbb{Z}_q$, then the identity is the only automorphism of $G$. This means that every channel is (trivially) automorphic-symmetric with respect to $\mathbb{Z}_q$.
\end{example}

Another example of automorphic-symmetric channels is generalized $q$-symmetric channels:

\begin{proposition}
\label{propSymIsAut}
Every generalized $q$-symmetric channel is automorphic-symmetric with respect to $(G,+)$.
\end{proposition}
\begin{proof}
Let $W:G\longrightarrow\mathcal{Y}$ be a generalized $q$-symmetric channel. Let $\mathcal{Y}_W$, $\pi_W$ and $p_W$ be as in Definition \ref{defGenqSym}.

Let $f:G\rightarrow G$ be an automorphism. Define $\pi_f:\mathcal{Y}\rightarrow\mathcal{Y}$ as $\pi_f=\pi_W\circ g_f\circ \pi_W^{-1}$, where $g_f:G\times\mathcal{Y}_W\rightarrow G\times\mathcal{Y}_W$ is defined as
$$g_f(x',y')=(f(x'),y').$$

Let $x\in G$ and $y\in\mathcal{Y}$. Define $(x',y')=\pi_W^{-1}(y)$. We have
\begin{align*}
W(\pi_f(y)|f(x))&=W\big((\pi_W\circ g_f\circ \pi_W^{-1})(y)\big|f(x)\big)\\
&=W(\pi_W(g(x',y'))|f(x))\\
&=W(\pi_W(f(x'),y')|f(x))\\
&\stackrel{(a)}{=}W(\pi_W(x',y')|x)=W(y|x),
\end{align*}
where (a) follows from the definition of generalized $q$-symmetric channels.
\end{proof}

\begin{definition}
Let $H$ be a subgroup of $G$. We say that $H$ is a \emph{characteristic subgroup} of $G$ if $f(H)=H$ for every automorphism $f$ of $G$. A subgroup that is not characteristic is said to be non-characteristic.

We denote the family of characteristic subgroups of $(G,+)$ by $\mathcal{H}_{ch}(G)$.
\end{definition}

In the rest of this section, we will show that automorphic-symmetric channels form an $\mathcal{H}_{ch}(G)$-polarizing family of channels.

\begin{lemma}
\label{lemAutPol1Step}
If $W:G\longrightarrow\mathcal{Y}$ is automorphic-symmetric, then $W^-$ and $W^+$ are automorphic-symmetric as well.
\end{lemma}
\begin{proof}
Let $f:G\rightarrow G$ be an automorphism and let $\pi_f:\mathcal{Y}\rightarrow\mathcal{Y}$ be a bijection satisfying $W(\pi_f(y)|f(x))=W(y|x)$. Define $\pi_f^-:\mathcal{Y}^2\rightarrow\mathcal{Y}^2$ and $\pi_f^+:\mathcal{Y}^2\times G\rightarrow\mathcal{Y}^2\times G$ as follows:
$$\pi_f^-(y_1,y_2)=(\pi_f(y_1),\pi_f(y_2)),$$
$$\pi_f^+(y_1,y_2,u_1)=(\pi_f(y_1),\pi_f(y_2),f(u_1)).$$

Obviously, $\pi_f^-$ and $\pi_f^+$ are bijections. Moreover, we have:
\begin{align*}
W^-&(\pi_f^-(y_1,y_2)|f(u_1))\\
&=W^-(\pi_f(y_1),\pi_f(y_2)|f(u_1))\\
&=\frac{1}{q}\sum_{u_2\in G}W(\pi_f(y_1)|f(u_1)+u_2)W(\pi_f(y_2)|u_2)\\
&\stackrel{(a)}{=}\frac{1}{q}\sum_{u_2\in G}W(\pi_f(y_1)|f(u_1)+f(u_2))W(\pi_f(y_2)|f(u_2))\\
&\stackrel{(b)}{=}\frac{1}{q}\sum_{u_2\in G}W(\pi_f(y_1)|f(u_1+u_2))W(\pi_f(y_2)|f(u_2))\\
&=\frac{1}{q}\sum_{u_2\in G}W(y_1|u_1+u_2)W(y_2|u_2)\\
&=W^-(y_1,y_2|u_1),
\end{align*}
where (a) and (b) follow from the fact that $f$ is an automorphism. This shows that $W^-$ is automorphic-symmetric. On the other hand, we have
\begin{align*}
W^+(&\pi_f^+(y_1,y_2,u_1)|f(u_2))\\
&=W^+(\pi_f(y_1),\pi_f(y_2),f(u_1)|f(u_2))\\
&=\frac{1}{q}W(\pi_f(y_1)|f(u_1)+f(u_2))W(\pi_f(y_2)|f(u_2))\\
&=\frac{1}{q}W(\pi_f(y_1)|f(u_1+u_2))W(\pi_f(y_2)|f(u_2))\\
&=\frac{1}{q}W(y_1|u_1+u_2)W(y_2|u_2)\\
&=W^+(y_1,y_2,u_1|u_2).
\end{align*}
This shows that $W^+$ is automorphic-symmetric as well.
\end{proof}

\begin{lemma}
\label{lemDetAut}
Let $\delta>0$. If $W$ is an automorphic-symmetric channel which is $\delta$-determined by a subgroup $H$ of $G$, then $W$ is $\delta$-determined by $f(H)$ for every automorphism $f:G\rightarrow G$.
\end{lemma}
\begin{proof}
Let $W:G\longrightarrow\mathcal{Y}$ be an automorphic-symmetric channel, let $f:G\rightarrow G$ be an automorphism, and let $H$ be a subgroup of $G$.

For every coset $A\in G/H$, define $$f(A)=\{f(x):\;x\in A\}.$$ It is easy to see that $f(A)\in G/f(H)$. Moreover, the reader can check that the mapping $f:G/H\rightarrow G/f(H)$ is an isomorphism of groups.

Now let $X$ be a uniformly distributed random variable in $G$ and let $Y$ be the output of the channel $W$ when $X$ is the input. For every $(x,y)\in G\times\mathcal{Y}$, we have $\displaystyle\mathbb{P}_{X,Y}(x,y)=\frac{1}{q}W(y|x)$. Therefore, for every $A\in G/H$, we have
\begin{equation}
\label{eqAutSym}
\begin{aligned}
&\mathbb{P}_{f^{-1}(X\bmod f(H)),\pi_f^{-1}(Y)}(A,y)\\
&\;\;=\mathbb{P}_{X\bmod f(H),Y}(f(A),\pi_f(y))=\sum_{x\in f(A)}\mathbb{P}_{X,Y}(x,\pi_f(y))\\
&\;\;=\sum_{x\in A}\mathbb{P}_{X,Y}(f(x),\pi_f(y))=\sum_{x\in A}\frac{1}{q}W(\pi_f(y)|f(x))\\
&\;\;\stackrel{(a)}{=}\sum_{x\in A}\frac{1}{q}W(y|x)=\sum_{x\in A}\mathbb{P}_{X,Y}(x,y)=\mathbb{P}_{X\bmod H,Y}(A,y),
\end{aligned}
\end{equation}
where (a) follows from the fact that $W$ is automorphic-symmetric. We deduce that
\begin{align*}
I(W[f(H)])&=I(X\bmod f(H);Y)\\
&\stackrel{(b)}{=}I(f^{-1}(X\bmod f(H));\pi_f^{-1}(Y))\\
&\stackrel{(c)}{=}I(X\bmod H;Y)=I(W[H]),
\end{align*}
where $(b)$ follows from the fact that $f:G/H\rightarrow G/f(H)$ and $\pi_f:\mathcal{Y}\rightarrow\mathcal{Y}$ are bijections. (c) follows from Equation \eqref{eqAutSym}.

Since $|G/f(H)|=|G/H|$ and $I(W[f(H)])=I(W[H])$, $W$ is $\delta$-determined by $H$ if and only if $W$ is $\delta$-determined by $f(H)$.
\end{proof}

\begin{proposition}
\label{propUnique}
There exists $\delta_0>0$ such that for every channel $W$ with input alphabet $G$, if $W$ is $\delta_0$-determined by a subgroup $H$ of $G$, then $H$ is the only subgroup $\delta_0$-determining $W$ (i.e., there is no subgroup $H'$ other than $H$ such that $W$ is $\delta_0$-determined by $H'$).
\end{proposition}
\begin{proof}
Define $\displaystyle\delta_0=\frac{1}{3}\log 2$. Assume that there are two subgroups $H_1$ and $H_2$ such that $W$ is $\delta_0$-determined by both $H_1$ and $H_2$.

Let $X$ be a random variable uniformly distributed in $G$ and let $Y$ be the output when $X$ is the input. We have
\begin{align*}
I(W[H_1])&=I(X\bmod H_1;Y)\\
&=H(X\bmod H_1)-H(X\bmod H_1|Y)\\
&=\log|G/H_1|-H(X\bmod H_1|Y).
\end{align*}
Therefore,
\begin{align*}
H(X\bmod H_1|Y)=\log|G/H_1|-I(W[H_1])\stackrel{(a)}{<}\delta_0,
\end{align*}
where (a) follows from the fact that $W$ is $\delta_0$-determined by $H_1$. Similarly, we can show that $H(X\bmod H_2|Y)<\delta_0$.

Now since there is a one-to-one correspondence between $(X\bmod H_1\cap H_2)$ and $(X\bmod H_1,X\bmod H_2)$, we have
\begin{align*}
H(X&\bmod H_1\cap H_2|Y)\\
&=H(X\bmod H_1,X\bmod H_2|Y)\\
&\leq H(X\bmod H_1|Y)+H(X\bmod H_2|Y)<2\delta_0.
\end{align*}
Therefore,
\begin{equation}
\label{eqeq1}
\begin{aligned}
I(W)&=I(X;Y)\geq I(X\bmod H_1\cap H_2;Y)\\
&=H(X\bmod H_1\cap H_2)-H(X\bmod H_1\cap H_2|Y)\\
&>\log|G/(H_1\cap H_2)|-2\delta_0\\
&=\log|G| - \log|H_1\cap H_2|-2\delta_0.
\end{aligned}
\end{equation}
Now since $W$ is $\delta_0$-determined by $H_1$, we have
\begin{align*}
I(W)-\log|G/H_1|<\delta_0,
\end{align*}
hence
\begin{equation}
\label{eqeq2}
I(W)<\log|G|-\log|H_1|+\delta_0.
\end{equation}
By combining Equations \eqref{eqeq1} and \eqref{eqeq2}, we get
$$\log\frac{|H_1|}{|H_1\cap H_2|}<3\delta_0=\log 2,$$
which implies that $\displaystyle|H_1\cap H_2|>\frac{|H_1|}{2}$. On the other hand, since $H_1\cap H_2$ is a subgroup of $H_1$, we have either $H_1\cap H_2=H_1$ or $|H_1\cap H_2|\leq\frac{1}{2}|H_1|$. Therefore, $H_1=H_1\cap H_2$ and so $H_1\subset H_2$. Similarly, we can show that $H_2\subset H_1$. Hence $H_1=H_2$.

We conclude that $W$ is $\delta_0$-determined by at most one subgroup of $G$.
\end{proof}

\begin{lemma}
\label{lemSecondProp}
Let $\delta_0$ be as in Proposition \ref{propUnique}. Automorphic-symmetric channels cannot be $\delta_0$-determined by a subgroup that is non-characteristic.
\end{lemma}
\begin{proof}
Let $W$ be an automorphic-symmetric channel. Assume that $W$ is $\delta_0$-determined by a non-characteristic subgroup $H$.

Since $H$ is non-characteristic, there exists an automorphism $f$ of $G$ such that $f(H)\neq H$. Lemma \ref{lemDetAut} implies that $W$ is $\delta_0$-determined by $f(H)$. This contradicts Proposition \ref{propUnique}.
\end{proof}

\begin{theorem}
\label{theMain}
Automorphic-symmetric channels form an $\mathcal{H}_{ch}(G)$-polarizing family of channels.
\end{theorem}
\begin{proof}
The theorem follows from Lemmas \ref{lemAutPol1Step} and \ref{lemSecondProp}.
\end{proof}

\vspace*{3mm}

Theorem \ref{theMain} shows that the synthetic channels of an automorphic-symmetric channel polarize to channels that are determined by characteristic subgroups.

\begin{corollary}
\label{corMain}
If $(G,+)$ does not contain any non-trivial characteristic subgroup, then the family of automorphic-symmetric channels is strongly polarizing.
\end{corollary}
\begin{proof}
The corollary follows from Theorem \ref{theMain}.
\end{proof}

\begin{corollary}
\label{corMain2}
If $(G,+)$ does not contain any non-trivial characteristic subgroup, then all automorphic-symmetric channels strongly polarize. In particular, all generalized $q$-symmetric channels strongly polarize.
\end{corollary}
\begin{proof}
The corollary follows from Corollary \ref{corMain}, Proposition \ref{propMain} and Proposition \ref{propSymIsAut}.
\end{proof}

\begin{example}
If $G\equiv\mathbb{F}_p^r$ for a prime $p$, then every non-trivial subgroup is non-characteristic. In this case, every automorphic-symmetric channel strongly polarizes.
\end{example}

\begin{example}
If $\displaystyle G=\prod_{i=1}^n \mathbb{F}_{p_i}^{r_i}$, where $p_1,\ldots,p_n$ are prime numbers, the reader can check that the characteristic subgroups of $(G,+)$ are those of the form $\displaystyle H=\prod_{i=1}^n\mathbb{F}_{q_i}^{l_i}$, where $l_i=0$ or $l_i=r_i$ for every $1\leq i\leq n$.

Therefore, if $\displaystyle G=\prod_{i=1}^n \mathbb{F}_{p_i}^{r_i}$ and $W$ is automorphic-symmetric, the polarization levels of $W$ are determined by subgroups of the form $\displaystyle H=\prod_{i=1}^n\mathbb{F}_{q_i}^{l_i}$, with $l_i=0$ or $l_i=r_i$ for every $1\leq i\leq n$.
\end{example}

\section{Discussion}
If $G\equiv\mathbb{Z}_q$ with composite $q$, then $G$ contains non-trivial characteristic subgroups, so we cannot apply Corollary \ref{corMain2}. Nevertheless, the simulations in \cite[Section V]{GulcuYeBarg} suggest that $q$-symmetric channels strongly polarize when the group $\mathbb{Z}_q$ is used. Proving this remains an open problem.

\section*{Acknowledgment}
I would like to thank Emre Telatar, Min Ye and Alexander Barg for helpful discussions.

\bibliographystyle{IEEEtran}
\bibliography{bibliofile}

\begin{thebibliography}{1}
\providecommand{\url}[1]{#1}
\csname url@samestyle\endcsname
\providecommand{\newblock}{\relax}
\providecommand{\bibinfo}[2]{#2}
\providecommand{\BIBentrySTDinterwordspacing}{\spaceskip=0pt\relax}
\providecommand{\BIBentryALTinterwordstretchfactor}{4}
\providecommand{\BIBentryALTinterwordspacing}{\spaceskip=\fontdimen2\font plus
\BIBentryALTinterwordstretchfactor\fontdimen3\font minus
  \fontdimen4\font\relax}
\providecommand{\BIBforeignlanguage}[2]{{%
\expandafter\ifx\csname l@#1\endcsname\relax
\typeout{** WARNING: IEEEtran.bst: No hyphenation pattern has been}%
\typeout{** loaded for the language `#1'. Using the pattern for}%
\typeout{** the default language instead.}%
\else
\language=\csname l@#1\endcsname
\fi
#2}}
\providecommand{\BIBdecl}{\relax}
\BIBdecl

\bibitem{Arikan}
E.~Ar{\i}kan, ``Channel polarization: A method for constructing
  capacity-achieving codes for symmetric binary-input memoryless channels,''
  \emph{Information Theory, IEEE Transactions on}, vol.~55, no.~7, pp.
  3051--3073, 2009.

\bibitem{SasogluTelAri}
E.~\c{S}a\c{s}o\u{g}lu, E.~Telatar, and E.~Ar{\i}kan, ``Polarization for
  arbitrary discrete memoryless channels,'' in \emph{Information Theory
  Workshop, 2009. ITW 2009. IEEE}, 2009, pp. 144--148.

\bibitem{SasS}
E.~\c{S}a\c{s}o\u{g}lu, ``Polar codes for discrete alphabets,'' in
  \emph{Information Theory Proceedings (ISIT), 2012 IEEE International
  Symposium on}, 2012, pp. 2137--2141.

\bibitem{ParkBarg}
W.~Park and A.~Barg, ``Polar codes for $q$-ary channels,'' \emph{Information
  Theory, IEEE Transactions on}, vol.~59, no.~2, pp. 955--969, 2013.

\bibitem{SahebiPradhan}
A.~G. Sahebi and S.~S. Pradhan, ``Multilevel channel polarization for arbitrary
  discrete memoryless channels,'' \emph{IEEE Transactions on Information
  Theory}, vol.~59, no.~12, pp. 7839--7857, Dec 2013.

\bibitem{RajTelA}
R.~Nasser and E.~Telatar, ``Polar codes for arbitrary {DMC}s and arbitrary
  {MAC}s,'' \emph{IEEE Transactions on Information Theory}, vol.~62, no.~6, pp.
  2917--2936, June 2016.

\bibitem{RajErgI}
R.~Nasser, ``An ergodic theory of binary operations, part {I}: Key
  properties,'' \emph{IEEE Transactions on Information Theory}, vol.~62,
  no.~12, pp. 6931--6952, Dec 2016.

\bibitem{RajErgII}
------, ``An ergodic theory of binary operations, part {II}: Applications to
  polarization,'' \emph{IEEE Transactions on Information Theory}, vol.~63,
  no.~2, pp. 1063--1083, Feb 2017.

\bibitem{GulcuYeBarg}
T.~C. Gulcu, M.~Ye, and A.~Barg, ``Construction of polar codes for arbitrary
  discrete memoryless channels,'' \emph{IEEE Transactions on Information
  Theory}, vol.~64, no.~1, pp. 309--321, Jan 2018.

\end{thebibliography}

\end{document}